\documentclass[12pt]{iopart}
\usepackage{amssymb}
\usepackage{amsthm}
\usepackage{braket}
\usepackage{hyperref}
\usepackage{setstack}
\usepackage{iopams}
\usepackage{xcolor}
\usepackage[numbers,sort&compress]{natbib}

\newtheorem{definition}{Definition}[section]
\newtheorem{theorem}{Theorem}[section]
\newtheorem{lemma}{Lemma}[section]

\theoremstyle{remark}
\newtheorem{remark}{Remark}[section]
\newtheorem{example}{Example}[section]

\newcommand{\R}{\mathbb{R}}
\newcommand{\abs}[1]{\left|#1\right|}
\newcommand{\norm}[1]{\left\|#1\right\|}

\newcommand{\comm}[2]{\left[#1,#2\right]}

\newcommand{\eqref}[1]{\eref{#1}}

\begin{document}
%\nocite{*}
\title[Robustness of quantum symmetries against  perturbations]{Robustness of quantum symmetries against  perturbations}

\author{Paolo Facchi$^{1,3}$, Marilena Ligab\`o$^2$ and Vito Viesti$^{1,3}$}

\address{$^1$Dipartimento di Fisica, Universit\`a  di Bari, I-70126 Bari, Italy\\$^2$Dipartimento di Matematica, Universit\`a  di Bari, I-70125 Bari, Italy\\$^3$INFN, Sezione di Bari, I-70126 Bari, Italy
}
\ead{vito.viesti@uniba.it}
\vspace{10pt}

\begin{abstract}
We investigate  quantum symmetries in terms of their large-time stability  with respect to perturbations of the Hamiltonian. We find a complete algebraic characterization of the set of symmetries robust against a single  perturbation and we use such result to characterize their stability with respect to arbitrary sets of perturbations.
\end{abstract}

\section*{Introduction}
Symmetries and conserved quantities play a crucial role in the dynamics of  physical systems~\cite{wigner_symmetry,symmetry}. 

 The fundamental connection between continuous symmetries and conservation laws revealed by Emmy Noether's theorem~\cite{Arnold,Fasano} is pivotal in understanding and constructing physical theories.

Conserved quantities can be characterized in terms of the generator of the evolution, i.e.\ the Hamiltonian of the system. In particular, if we focus on quantum systems,  it is well known that the conserved quantities are the operators that commute with the Hamiltonian~\cite{Sakurai}. Therefore, given a Hamiltonian, it is possible, at least in principle, to find all its symmetries. 

However, in the description of a physical system one usually neglects some interactions. Then the Hamiltonian used in the model differs by a small perturbation from the real one.
 A crucial requirement of a proper physical model is that it should be robust against perturbations. For example, if the dynamics of a physical system is modelled by differential equations it is of the utmost importance that a continuous dependence of the solutions with respect to the coefficients and /or the initial conditions be established. 
 
 Such robustness plays also a pivotal role from a practical point of view, for instance, in analog quantum simulations, where
one tries to mimic the  dynamics of the system of interest on another quantum system~\cite{simul2}. Clearly, it is impossible to avoid errors and imperfections in the system one uses to simulate the model. It is then essential to study the stability of the prediction of the model with respect to perturbations of the Hamiltonian, which describe these unavoidable imperfections of the real world hardware~\cite{simul,qsimul}.

The question that we address in this work is the following: what happens to a conserved quantity under small perturbations of the Hamiltonian? Is it still (approximately) conserved? 

As one expects, in general the answer is negative. Indeed, a generic perturbation does not commute with the symmetry of the unperturbed system.  However, if we consider small times, by  standard perturbation theory, the perturbed evolution can be shown to drive a conserved quantity very close to its initial value: so it is true that even if the quantity is no longer conserved, at least it shows a small-time robustness against small perturbations~\cite{kam,onebound}.

At large times this behavior is no longer true generally: it is possible that the small effect of the perturbation accumulates over time and eventually a symmetry is drifted far away from its initial value.  However, there are some special conserved quantities that, despite the perturbation, show a stability against the perturbed evolution that remains uniform in time. We will call them robust symmetries; if this is not the case we will call them fragile. It is crucial to find out what  the robust symmetries are, since they are the relevant ones from a physical perspective.

This problem is in close analogy to the Kolmogorov-Arnold-Moser (KAM) theory of classical mechanics~\cite{Arnold}. The most famous example  is the study of the motion of the planets in the solar system.  In the description of the motion of a planet, as a first approximation one usually considers only its interaction
with the sun,  neglecting the presence of all the other planets.  It is well known that in such a case the planet moves along closed elliptical orbits. However, the interactions with the other planets,  despite negligible with respect to the sun, are not zero: and their effect could accumulate over long times.  KAM theorems ensure the existence of large-time stable planetary orbits along stable tori~\cite{Dumas}. Robust symmetries in quantum mechanics are the quantum analogues of those KAM tori.

In this paper we establish a classification of robust symmetries: we study the stability with respect to an arbitrary set of perturbations. In particular, we will answer the following question: given a (relatively bounded) perturbation, what are the symmetries uniformly stable against the perturbed dynamics?

We will see that a perturbation induces a family of subprojections of the spectral projections of the unperturbed Hamiltonian. In Theorem~\ref{main} we will prove that the robust symmetries are exactly the ones that commute with these subprojections.

After analyzing the single-perturbation stability, we will use this result to see what are the symmetries robust against an arbitrary set of perturbations $\mathcal{P}$. A very interesting physical case is when the set of perturbations is made by the operators that commute with a certain symmetry of the Hamiltonian. If for instance we deal with a system that is rotationally symmetric, one could require that also the perturbations should satisfy the same requirement. In Theorem~\ref{thm:restricted} we see that in such a case the set of robust symmetries has a nice algebraic structure.

Then, we will see what are the symmetries robust against all possible (relatively bounded) perturbations: in Theorem~\ref{complete} we will show that they are the bounded functions of the Hamiltonian, a result proved for finite-dimensional quantum systems~\cite{kam, eternal}.
Here we will extend this result to an infinite-dimensional Hilbert space, and prove its validity for unbounded Hamiltonians.

An important difference between the finite- and the infinite-dimensional case arises
in the wandering range of robust symmetries, i.e.\ the deviation of their perturbed evolution from their unperturbed value. 
By constructing explicit examples, we show in fact that in an infinite-dimensional Hilbert space the wandering range of a robust symmetry can be   $O(\varepsilon^\gamma)$ with arbitrarily small positive $\gamma$, at striking variance with the finite-dimensional case where the wandering range is always $O(\varepsilon)$, as the strength of the perturbation vanishes, $\varepsilon \to 0$.

Finally, in Theorem \ref{thm:adinv} we show how conserved quantities of the perturbed Hamiltonian can be constructed  starting from a robust symmetry. Indeed, when the Hamiltonian is perturbed, a symmetry is no longer conserved. However, if the conserved quantity is robust, it is possible to continuously deform it into a conserved quantity of the perturbed dynamics. In other words if a symmetry is robust, it is not lost but only bent by the perturbation. We will call these deformed symmetries, quantum adiabatic invariants, in analogy with classical adiabatic invariants~\cite{Arnold}.

A related approach was considered in~\cite{gorini}: here equilibrium states are defined as the density operators which are quantum adiabatic invariants with respect to all possible perturbations of the Hamiltonian. Under some additional hypotheses on the Hamiltonian, they find that the density operator must be a function of the Hamiltonian, in accordance with our general findings. This allows to characterize  Gibbs states 
as the only states robust against perturbations and then as equilibrium states~\cite{stability}.

More generally, it is not difficult to see how our results could be relevant in quantum thermodynamics~\cite{Calabrese},
where symmetries play a crucial role, as in the investigation of the dynamics of quantum quenches~\cite{glimmers}, or in the description of thermalization processes~\cite{Mori_2018} and of ergodicity breaking~\cite{ergodicity1, ergodicity2}.

The paper is organized as follows: in section ~\ref{sec:1} we briefly recall the description of the quantum evolution in the Heisenberg picture, the definition of conserved quantities and their strict connection with continuous symmetries. In section~\ref{sec:2} we give the definition of fragile and robust symmetries against  a given set of perturbations. Section~\ref{sec:3} is the core of the paper: we give the complete algebraic characterization of symmetries robust against a fixed perturbation. In section~\ref{sec:4} we characterize the symmetries robust against an arbitrary set of perturbations, and look at their algebraic structure in some interesting situations.
 Then, in section~\ref{sec:adiabatic} we look at the wandering range of robust symmetries and at their connection with adiabatic invariants. Finally, in Sec.~\ref{sec:conclusions} we draw our conclusions.

\section{Notation and preliminaries}
\label{sec:1}
Let $\mathcal{H}$ be a complex separable Hilbert space. We focus our attention on closed quantum systems: in such a case the evolution of the system is described by a one-parameter strongly continuous unitary group $\{\rme^{-itH}\}_{t \in \R}$, where the Hamiltonian operator $H$ is the generator of the group and is a (possibly unbounded) self-adjoint operator with domain $D(H) \subseteq \mathcal{H}$~\cite{Reed}. We will use the Heisenberg picture, where the observables evolve in time and the states remain fixed. More explicitly, the time evolution of any bounded operator $A \in B(\mathcal{H})$ is described by the following law:
\begin{equation}
	t \in \R \mapsto A_t:=\rme^{{\rmi}tH}A\rme^{-{\rmi}tH} \in B(\mathcal{H}),
	\label{eq:evolution}
\end{equation}
where we denote with $A_t$ the evolution of the operator $A$ at time $t$~\cite{Teschl}.

There are observables that play a privileged role in the description of a system:  conserved quantities. They are the operators that, despite the evolution described in~\eqref{eq:evolution}, do not depend on time. Because of  Noether's theorem, there is a one to one correspondence between conserved quantities and continuous symmetries of a system: we are going to use the terms ``conserved quantities'' and ``symmetries'' interchangeably.
\begin{definition}
We say that $S \in  B(\mathcal{H})$ is a symmetry of the Hamiltonian $H$ if for all $t \in \R$:
\begin{equation}
\rme^{{\rmi}tH}S\rme^{-{\rmi}tH}=S.
\end{equation}
This is equivalent to the vanishing of the commutator: $[S, \rme^{{\rmi}tH}]:=S\rme^{{\rmi}tH}- \rme^{{\rmi}tH}S= 0$, for all times $t\in\R$. Thus, the set of all the symmetries of the Hamiltonian $H$ coincides with the commutant of $H$, denoted with $\{H\}'$. 
\end{definition}

\begin{remark}
\label{rmk:commutant}
We recall that given a set $\mathcal{C}$ of bounded operators, its commutant is given by
	\begin{equation}
	  \mathcal{C}'=\{A \in B(\mathcal{H}) : 
		[A, C]=0,  \text{ for all } C\in \mathcal{C}  \}.
		\end{equation}
	The commutant of a (not necessarily bounded) self-adjoint operator $H$ is the von Neumann algebra of all bounded operators which commute with the family $\{P_H(\Omega)
	\}$ of the spectral projections   of $H$:
	\begin{equation}
		  \{H\}'= \{ P_H(\Omega) \,:\, \Omega \subseteq \R,\; \Omega \text{ measurable} \}' \, .
	\end{equation}
	See e.g.~\cite[Prop.14.3.5]{oliv}. This definition naturally extends to many self-adjoint operators. 
	
	Notice that the commutant of $H$ is  virtually independent of its spectrum and does not change under a reparametrization of the spectrum, as far as it is one to one.  In other words, $\{H\}'=\{f(H)\}'$ for all $f:\R \to \R$ one to one.
	Moreover, $\{H\}'$ has a privileged Abelian subalgebra, the bicommutant of $H$, which is generated by the spectral projections of $H$:
	 \begin{equation}
		\{H\}''=\{P_H(\Omega) :\Omega \text{ measurable subset of } \R  \}'' \subset \{H\}'.
		\label{eq:H''def}
	\end{equation}
	By von Neumann's bicommutant theorem, $\{H\}''$ is the Abelian von Neumann algebra of the measurable bounded functions of $H$. Thus, since $\{H\}'=\{H\}'''$, the commutant of $H$  is the von Neumann algebra of bounded operators commuting with all bounded Borel functions of $H$, namely,
	\begin{equation}
 \{H\}'= 
		\{f(H) \,:\; f\!:\!\R \to\R, \; f\text{ bounded Borel}   \}' \, .
	\end{equation}
	\end{remark}
\vspace{2mm}

Our objective is to introduce a refined classification of quantum symmetries $S\in\{H\}'$, distinguishing them between \emph{fragile} and \emph{robust} in relation to their time evolution with respect to perturbed Hamiltonians. 

Let us introduce a class of admissible perturbations (that preserve the self-adjointness of the Hamiltonian $H$).
\begin{definition}
A linear operator $V$ on $\mathcal{H}$ is \emph{$H$-bounded} if 
\begin{itemize}
\item $D(H) \subset D(V)$; 
\item there are nonnegative constants $a, b \geq 0$ such that for all $\psi \in D(H)$:
\begin{equation}\label{H-bound}
\| V \psi \| \leq a  \| H \psi\| + b \| \psi \|.
\end{equation}
\end{itemize}
The greatest lower bound $a_V$ of all possible $a$ in (\ref{H-bound}) is called \emph{$H$-bound of $V$}. 
\end{definition}
The Kato-Rellich theorem ensures the self-adjointness of the perturbed Hamiltonian $H+ \varepsilon V$, if $V $ is a symmetric $H$-bounded operator and $\varepsilon \in \R$ is small enough~\cite{Kato_per}.
\begin{theorem}[Kato-Rellich]
\label{th:perturbation}
Let $V$ be a symmetric and $H$-bounded operator with $H$-bound $a_V$. Then for all $\varepsilon \in \R$, with $\abs{\varepsilon} a_V <1$, the operator $H+\varepsilon{V}$
is self-adjoint on $D(H)$. 
\end{theorem}

Obviously, if $V$ is bounded (which is always the case for a finite-dimensional Hilbert space $\mathcal{H}$), then it has $H$-bound 0 and Kato-Rellich holds for all $\varepsilon\in\R$.

\section{$\mathcal{P}$-robustness}\label{sec:2}
Consider a set $\mathcal{P}$ of symmetric operators with $H$-bound less than 1. Then, $H+\varepsilon V$ is self-adjoint for all $|\varepsilon|<1$ and $V\in\mathcal{P}$. 
We want to classify the symmetries in terms of their large-time evolution with respect to the family of perturbed Hamiltonians 
\begin{equation}\label{eqnHe}
H(\varepsilon):=H+\varepsilon V,  \qquad \varepsilon \in (-1,1),
\end{equation}
with $V \in \mathcal{P}$, according to the following definition. 
\begin{definition}\label{def:Pfr} Let $S \in \{H\}'$ be a symmetry of the Hamiltonian $H$, and let the \emph{set $\mathcal{P}$ of perturbations of $H$}  be a set of $H$-bounded symmetric operators.
\begin{itemize}
\item $S$ is \emph{$\mathcal{P}$-robust} if for all $V\in\mathcal{P}$ and for all $\psi\in\mathcal{H}$:
\begin{equation}
    \sup_{t\in\mathbb{R}}\norm{\left(\rme^{{\rmi}t(H+ \varepsilon V)}S\rme^{-{\rmi}t(H+ \varepsilon V)}-S\right)\psi} \to 0, \quad \text{as } \varepsilon \to 0.
\end{equation} 
The set of all $\mathcal{P}$-robust symmetries is denoted with $\mathcal{R}_{\mathcal{P}}(H)$. It contains the set $\hat{\mathcal{R}}_{\mathcal{P}}(H) $ of  \emph{$\mathcal{P}$-unbroken} 
 symmetries, for which $\rme^{{\rmi}t(H+ \varepsilon V)}S\rme^{-{\rmi}t(H+ \varepsilon V)}=S$.
\item $S$ is \emph{$\mathcal{P}$-fragile} if it is not $\mathcal{P}$-robust, i.e. if there exists a perturbation $V\in\mathcal{P}$ and a vector $\psi\in\mathcal{H}$ such that
\begin{equation}
    	\limsup_{\varepsilon\rightarrow0} \left[\sup_{t\in\mathbb{R}}\norm{\left(\rme^{{\rmi}t(H+ \varepsilon V))}S\rme^{-{\rmi}t(H+ \varepsilon V)}-S\right)\psi} \right]>0.
	\label{eq:Pfragile}
\end{equation}
\item  If $\mathcal{P}$ contains only one element, namely $\mathcal{P}=\{V\}$, we say that $S$ is \emph{$V$-robust}. If $S$ is $V$-robust for all possible symmetric $H$-bounded perturbations $V$, we say that $S$ is a \emph{completely robust} symmetry. The set of all  completely robust symmetries is denoted with $\mathcal{R}(H)$.
\end{itemize}
\end{definition}
In other words, a symmetry $S$ 
is $\mathcal{P}$-robust if  $\rme^{{\rmi}t(H+ \varepsilon V)}S\rme^{-{\rmi}t(H+ \varepsilon V)} \psi$ remains close, for small $\varepsilon$, to its unperturbed value $\rme^{{\rmi}t H}S\rme^{-{\rmi}t H} \psi=S\psi$ for every time $t\in \R$, for every possible perturbation $V \in \mathcal{P}$ and for every possible state $\psi \in \mathcal{H}$,
that is
\begin{equation}
	\rme^{{\rmi}t(H+ \varepsilon V)}S\rme^{-{\rmi}t(H+ \varepsilon V)} \approx S, \quad \text{for } \abs\varepsilon\ll 1, \quad \text{ uniformly in time $t$}.
\end{equation}
In particular, $S$ is $\mathcal{P}$-unbroken if it remains a symmetry of all the perturbed Hamiltonians, namely $\rme^{{\rmi}t(H+ \varepsilon V)}S\rme^{-{\rmi}t(H+ \varepsilon V)}=S$ for all times $t\in\R$ and perturbations $V\in\mathcal{P}$.

On the other hand, $S$ is $\mathcal{P}$-fragile if, however small  $\varepsilon$ is, there is a  perturbation $V \in \mathcal{P}$ and a  vector $\psi \in \mathcal{H}$ such that  $\rme^{{\rmi}t(H+ \varepsilon V)}S\rme^{-{\rmi}t(H+ \varepsilon V)} \psi$ drifts away from $S\psi$ and their distance accumulates over time.  

In the next section we will study in detail the robustness of a symmetry against a single perturbation.
\section{Robustness against a single perturbation}
\label{sec:3}
In order to perform our analysis on the symmetries, we introduce an assumption on the unperturbed self-adjoint Hamiltonian $H$: we assume that $H$ has \emph{compact resolvent}, i.e. $(H-z)^{-1}$ is a compact operator for some $z \in \mathbb{C}$. In such a case the spectrum  of $H$ consists entirely of isolated eigenvalues with finite multiplicity and its spectral decomposition reads
    \begin{equation}
            H\psi=\sum_{k\ge1}h_kP_k\psi, \qquad \forall\psi\in{D(H)},
            \label{eq:specresH}
    \end{equation}
 where  $\{h_k\}_{k \geq 1}\subseteq\R$ are the distinct eigenvalues of $H$ and  $\{P_k\}_{k \geq 1}$ are its finite-rank eigenprojections:
\begin{equation}
P_k^\dagger=P_k, \quad P_kP_l=\delta_{kl}P_k, \quad \forall k,l \geq 1, \qquad 
\sum_{k\ge1}P_k\psi=\psi, \quad \forall\psi\in\mathcal{H}.
\end{equation}
A paradigmatic example that can be kept in mind is the Hamiltonian of the ($n$-dimensional isotropic) harmonic oscillator. Obviously, if the Hilbert space $\mathcal{H}$ is finite dimensional all Hamiltonians have compact resolvent.

Now we recall an important result on the spectral properties of the family  $\{H(\varepsilon)\}_{\varepsilon \in (-1,1)}$ of perturbations of $H$~\cite{Kato_per}.

\begin{theorem}[Kato] \label{th:KATO}
\begin{enumerate}
Let $H$ be a self-adjoint operator with compact resolvent and spectral resolution~\eqref{eq:specresH}.
Let $V$ be a symmetric operator with $H$-bound less than 1. Then,
\item for all $\varepsilon \in (-1,1)$, the perturbed Hamiltonian  $H(\varepsilon)$   in (\ref{eqnHe}) has compact resolvent and its spectral decomposition reads
    \begin{equation}
    H(\varepsilon)\psi=\sum_{n\ge1}h_n(\varepsilon)P_n(\varepsilon)\psi,\qquad\forall\psi\in{D(H)},
    \label{eq:spectral}
\end{equation}
 where  $\{h_n(\varepsilon)\}_{n \geq 1}$ are the eigenvalues of $H(\varepsilon)$ and $\{P_n(\varepsilon)\}_{n \geq 1}$ are its finite-rank eigenprojections;
\item for all $n \geq 1$, the maps $\varepsilon \in (-1,1) \mapsto h_n(\varepsilon) \in \R$, and $\varepsilon \in (-1,1) \mapsto P_n(\varepsilon) \in{B}(\mathcal{H})$, 
are analytic, with $h_n\neq h_m$ for $n\neq m$;
\item the family $\{P_n(0)\}_{n \geq 1}$ is a family of subprojections of  $\{P_k\}_{k \geq 1}$, namely for all $n \geq 1$ there is a unique $k \geq 1$ such that 
$P_n(0) P_k= P_k P_n(0)= P_n(0)$,
so that $\textrm{Range}(P_n(0)) \subseteq \textrm{Range}(P_k)$.
\end{enumerate}
\end{theorem}

We are ready to give a complete algebraic characterization of all the $V$-robust symmetries.

\begin{theorem}
\label{main}
    Let $H$ be a self-adjoint compact-resolvent operator and $V$ be symmetric and $H$-bounded. 
    Consider a symmetry $S\in\{H\}'$. Then
    \begin{enumerate}
    \item $S$ is $V$-robust if and only if 
    \begin{equation}
    	\quad\comm{S}{P_{n}(0)}=0, \quad \text{for all } n \geq 1,
    \end{equation}
    where $\{P_n(\varepsilon)\}_{n \geq 1}$ are the eigenprojections of the perturbed Hamiltonian $H+\varepsilon V$;
    \item if $V$ commutes with $H$ and is self-adjoint, then $S$ is $V$-robust if and only if $S$ and $V$ commute, i.e.\ $S$ is $V$-unbroken. 
    \end{enumerate}
    Furthermore, the set of $V$-robust symmetries
    \begin{equation}
    	\quad\mathcal{R}_{\{V\}}(H)=\{P_n(0): n \geq 1\}'
    \end{equation}
   is a von Neumann algebra with 
   \begin{equation}
   	\quad\{H\}''\subseteq  \mathcal{R}_{\{V\}}(H)\subseteq \{H\}'. 
   \end{equation}
   and, if $V=V^\dagger$,
   \begin{equation}
   	\quad\{H,V\}'\subset\mathcal{R}_{\{V\}}(H) .
   \end{equation}
    \end{theorem}
\begin{remark}
In the case under consideration of a pure-point Hamiltonian $H$ the bicommutant $\{H\}''$ in~\eqref{eq:H''def} reduces to the von Neumann algebra generated by the eigenprojections of $H$, namely,
\begin{equation}
	\{H\}'' = \{P_k : k\geq 1\}''
	= \Bigl\{\sum_{k \geq 1} a_{k} P_{k}  \,:\, (a_{k})_{k \geq 1} \in \ell^{\infty}\Bigr\} ,
\label{eq:bicommutantCR}
\end{equation}
where the convergence of the series is in the strong topology.

The symmetries of the Hamiltonian $H$ are the bounded operators which commute with the Hamiltonian $H$, or equivalently with  (the Abelian von Neumann algebra  $\{H\}''$ generated by) the family of its eigenprojections $\{P_k\}_{k \geq 1}$. On the other hand, the $V$-robust symmetries are the ones which commute with (the larger Abelian von Neumann algebra generated by) the subprojections $\{P_n(0)\}_{n \geq 1}$ of $\{P_k\}_{k \geq 1}$ induced by~$V$.

The robust symmetries form a von Neumann algebra: thus the product and the linear combination of robust symmetries is still robust and so is the (strong) limit of robust symmetries.
\end{remark}

Theorem~\ref{main} is a corollary, by Kato's theorem~\ref{th:KATO}, of the following more general result on the stability of symmetries against continuous self-adjoint deformations of $H$.

\begin{theorem} \label{thm:main0}
Let $H$ be a pure-point self-adjoint operator.
	Let $I$ be a real neighbourhood of $0$. Let $\varepsilon\mapsto H(\varepsilon) $ be a continuous deformation of $H$, where $H(\varepsilon)$ is self-adjoint on $D(H(\varepsilon))=D(H)$ for all $\varepsilon\in I$,  and $H(0)=H$. Assume that
	\begin{enumerate}
		\item $H(\varepsilon) \psi = \sum_{n\geq 1}h_n(\varepsilon) P_n(\varepsilon) \psi$, for all $\psi\in D(H)$ and all $\varepsilon\in I$;
		\item $\varepsilon \mapsto h_n(\varepsilon)$ are continuous real-valued functions, with $h_n(\varepsilon)\neq h_m(\varepsilon)$ for $n\neq m$ and all $\varepsilon \in I\setminus\{0\}$;
		\item $\{P_n(\varepsilon)\}_{n\geq 1}$ is a complete orthogonal family of projections, with  $P_n(\varepsilon)= U(\varepsilon)P_n(0)U(\varepsilon)^\dag$ for all $\varepsilon\in I$, where $\varepsilon\mapsto U(\varepsilon)$ is strongly continuous and $U(\varepsilon)$ is unitary.
	\end{enumerate}
	
	Let $S$ be a symmetry of $H$, i.e.\ $S\in\{H\}'$. Then, 
	\begin{equation}
    \sup_{t\in\mathbb{R}}\norm{\left(\rme^{{\rmi}t H(\varepsilon)}S\rme^{-{\rmi}tH(\varepsilon)}-S\right)\psi} \to0, \quad \text{as } \varepsilon \to 0, 
    \label{eq:robustS}
    \end{equation} 
for all $\psi\in\mathcal{H}$, if and only if
\begin{equation}
	S\in\{P_n(0): n\geq 1\}' .
	\label{eq:SP_n(0)}
\end{equation}

When~\eqref{eq:robustS} holds, we say that the symmetry $S$ is \emph{robust against the deformation $H(\varepsilon)$}.

	\end{theorem}
\begin{remark}
Here we consider deformations $H(\varepsilon)$ of a pure-point Hamiltonian $H$ (not assumed to have a compact resolvent), which are not necessarily of the form~\eqref{eqnHe}, whence $H(\varepsilon)$ are no longer guaranteed to have a compact resolvent. However,  their spectrum is still required to be pure point, but the eigenvalues can have finite accumulation points and infinite degeneracy. Finally, the eigenvalues and the eigenprojections are only required to be continuous. 
\end{remark}

The proofs of Theorems~\ref{main} and~\ref{thm:main0} are postponed to Subsection~\ref{susect:proof123}, after proving two preliminary lemmas.  The first one deals with the construction of an eternal block-diagonal approximation of the  deformed Hamiltonian $H(\varepsilon)$ in Theorem~\ref{thm:main0}, which generates a unitary group commuting with the one generated by $H$ (and then block-diagonal), and approximating the group generated by $H(\varepsilon)$ uniformly in time (eternally). The second lemma regards the  splitting of the perturbed evolution of a symmetry in two parts: a fragile component and a robust one.

\subsection{Eternal Block-Diagonal Approximation}
Consider  a continuous deformation $\varepsilon\mapsto H(\varepsilon)$ of a pure-point self-adjoint operator $H$,  as in Theorem~\ref{thm:main0}. For all $ \varepsilon \in I$ define the self-adjoint operator 
    \begin{equation}\label{eqn:Htildee}
        \tilde{H}(\varepsilon):=
        U(\varepsilon)^\dagger H(\varepsilon) U(\varepsilon) =
        \sum_{n\ge1}h_{n}(\varepsilon)P_{n}(0), 
    \end{equation}
    on the domain $U(\varepsilon)^\dagger D(H)$,
 and consider the unitary group $t\mapsto \rme^{-\rmi t \tilde{H}(\varepsilon)}$ it generates. We now show that this group is an approximation of $t\mapsto \rme^{-\rmi t H(\varepsilon)}$ uniform in time.
  
\begin{lemma}\label{lemma:123}
Let $\varepsilon\mapsto H(\varepsilon)$ be a continuous deformation of a pure-point self-adjoint operator $H$, as in Theorem~\ref{thm:main0}. Then, 
 for all $\psi\in\mathcal{H}$,
    \begin{equation}
\sup_{t \in \mathbb{R}}\norm{\left(\rme^{-{\rmi}t{H}(\varepsilon)}-\rme^{-{\rmi}t\tilde{H}(\varepsilon)}\right)\psi}\to 0, \quad \text{as } \varepsilon\to0.
    \label{eq:upsi}
    \end{equation}
\end{lemma}
\begin{proof} 
Let $\psi\in\mathcal{H}$ and $t \in \mathbb{R}$, we have that
\begin{eqnarray}
        \norm{(\rme^{-{\rmi}tH(\varepsilon)}-\rme^{-{\rmi}t\tilde{H}		(\varepsilon)})\psi} &= &
        \norm{(U(\varepsilon)\rme^{-{\rmi}t\tilde{H}(\varepsilon)}U(\varepsilon)^\dagger-\rme^{-{\rmi}t\tilde{H}(\varepsilon)})\psi}\nonumber \\
        &= &\norm{\comm{\rme^{-{\rmi}t\tilde{H}(\varepsilon)}}{U(\varepsilon)^\dagger-\mathbb{I}}\psi}\nonumber \\ 
        &\leq & \norm{(U(\varepsilon)-\mathbb{I})\psi}+\norm{(U(\varepsilon)-\mathbb{I})\rme^{-{\rmi}t\tilde{H}(\varepsilon)}\psi}.
\label{eq:strong}
    \end{eqnarray}
By the property $(iii)$ in Theorem~\ref{thm:main0} we have that
   $\norm{(U(\varepsilon)-\mathbb{I})\psi}\to 0$, as $\varepsilon\to 0$.
        Let us consider the second term of~\eqref{eq:strong} and the resolution of the identity applied to $\psi$:
    \begin{equation*}
        \psi=\sum_{n\ge1}P_n(0)\psi.
    \end{equation*}
For all $N\in\mathbb{N^*}$,  we can write
    \begin{equation}
        \psi=\psi_{\leq{N}}+\psi_{>N}\label{eq:decN},
    \end{equation}
    where
\begin{equation}
    \psi_{\leq{N}}=\sum_{n=1}^NP_n(0)\psi, \qquad 
    \psi_{>N}=\sum_{n=N+1}^{+ \infty}P_n(0)\psi,
\end{equation}
and
\begin{equation} \label{eqn:psiN}
    \norm{\psi_{>N}}\rightarrow0 ,\quad \text{as }{N\rightarrow\infty}.
\end{equation}
By plugging~\eqref{eq:decN} into the second term of~\eqref{eq:strong}, we get:
\begin{eqnarray*}
    \norm{(U(\varepsilon)-\mathbb{I})\rme^{-{\rmi}t\tilde{H}(\varepsilon)}\psi} &=\norm{(U(\varepsilon)-\mathbb{I})\rme^{-{\rmi}t\tilde{H}(\varepsilon)}(\psi_{\leq{N}}+\psi_{>N})}\\&\leq
    \norm{(U(\varepsilon)-\mathbb{I})\rme^{-{\rmi}t\tilde{H}(\varepsilon)}\psi_{\leq{N}}} +2\norm{\psi_{>N}}     \\&\leq
        \norm{(U(\varepsilon)-\mathbb{I})\rme^{-{\rmi}t\tilde{H}(\varepsilon)}\sum_{n=1}^{N}P_n(0)\psi}  +2\norm{\psi_{>N}} \\&=\norm{\sum_{n=1}^{N}\rme^{-{\rmi}th_n(\varepsilon)}(U(\varepsilon)-\mathbb{I})P_n(0)\psi} +2\norm{\psi_{>N}}
\end{eqnarray*}
Let $\eta>0$, then by (\ref{eqn:psiN}) there is $N_0 >0$  such that 
\begin{equation*}
    \norm{\psi_{>N_0}}<\frac{\eta}{6},
\end{equation*}
Moreover, by the strong continuity of $\varepsilon\mapsto U(\varepsilon)$, there is $\varepsilon_0>0$ such that for all $\varepsilon \in I$, with $\abs{\varepsilon}<\varepsilon_0$, and for all $n=1,\dots,N_0$:
\begin{equation*}
    \norm{(U(\varepsilon)-\mathbb{I})P_n(0)\psi}<\frac{\eta}{3N_0}, 
\end{equation*}
and 
\begin{equation*}
 \norm{(U(\varepsilon)-\mathbb{I})\psi}<\frac{\eta}{3}.
\end{equation*}
Then for all $|\varepsilon| < \varepsilon_0$:
\begin{eqnarray*}    
        \norm{(\rme^{-{\rmi}tH(\varepsilon)}-\rme^{-{\rmi}t\tilde{H}		(\varepsilon)})\psi} &\leq &
        \norm{(U(\varepsilon)-\mathbb{I})\psi}+\norm{(U(\varepsilon)-\mathbb{I})\rme^{-{\rmi}t\tilde{H}(\varepsilon)}\psi} \\
        &< & \frac{\eta}{3} + \frac{\eta}{3}+\frac{\eta}{3} = \eta.
            \end{eqnarray*}
    This proves the limit~\eqref{eq:upsi}.
 \end{proof}
\begin{remark}
The proof of  
Lemma~\ref{lemma:123} is based on the property 
\begin{equation}
    \sup_{t\in \mathbb{R}}\norm{\left(U(\varepsilon)-\mathbb{I}\right)\rme^{-{\rmi}t\tilde{H}(\varepsilon)}\psi}\to0, \quad \text{as } \varepsilon\to 0, \qquad \forall \psi \in \mathcal{H}.
    \label{eq:upsi1}
\end{equation}
    For this property to hold the assumption of  pure-point spectrum of $H(\varepsilon)$, and thus of $\tilde{H}(\varepsilon)$, is crucial. Indeed, if the Hamiltonian has a pure-point spectrum, the orbit of a generic $\psi\in\mathcal{H}$ has a compact closure (see e.g.~\cite{oliv}). In such a case a vector, during its evolution, spends most of the time, apart a small error uniform in time, in a finite-dimensional manifold. Then the operation of taking the supremum over time does not span the full Hilbert space, but just a finite-dimensional submanifold, where~\eqref{eq:upsi1} is trivially satisfied because of the strong continuity of $\{U(\varepsilon)\}_{\varepsilon\in I}$.
     On the other hand, when the spectrum is not pure point such property is not in general satisfied.
\end{remark}
\subsection{Fragile and robust components of a symmetry}
In the following lemma we split the evolution of a symmetry with respect to the perturbed dynamics in two parts: a fragile component and a robust one.
\begin{lemma}
\label{l:evolution} 
Let $H$ and $H(\varepsilon)$ be as in Theorem~\ref{thm:main0} and $\tilde{H}(\varepsilon)$ be as in~\eqref{eqn:Htildee}. Let $S \in \{H\}'$ be a symmetry of the Hamiltonian $H$. Define for all $\varepsilon \in I$ and $t \in \mathbb{R}$    
\begin{eqnarray}
 A(t,\varepsilon):=\rme^{{\rmi}t{H}(\varepsilon)}\comm{S}{\rme^{-{\rmi}t\tilde{H}(\varepsilon)}},
 \label{eq:Arob}\\
 B(t,\varepsilon):=\rme^{{\rmi}t{H}(\varepsilon)}\comm{S}{\rme^{-{\rmi}tH(\varepsilon)}-\rme^{-{\rmi}t\tilde{H}(\varepsilon)}},
 \label{eq:Bfrag}
\end{eqnarray}
so that
\begin{equation}
\rme^{{\rmi}tH(\varepsilon)}S\rme^{-{\rmi}tH(\varepsilon)}-S=A(t,\varepsilon)+B(t,\varepsilon).
 \end{equation}
Then for all $\psi\in\mathcal{H}$,
\begin{equation}\label{eq:robust}
\sup_{t \in \mathbb{R}}\, \norm{B(t,\varepsilon)\psi}\to 0, \quad \text{as } \varepsilon\to0.
\end{equation}
\end{lemma}
\begin{remark}
As Eq.~\eqref{eq:robust} makes clear, the operator $B(t,\varepsilon)$ describes a robust component of the perturbed evolution of $S$. Indeed, its contribution is negligible, uniformly in time, for $\varepsilon$  sufficiently small. On the other hand, as we will see,
$A(t,\varepsilon)$ is responsible for a possible non-negligible divergence of the symmetry from its initial value. For this reason it represents a fragile component of the evolution of $S$.
\end{remark}
 Let us prove Lemma~\ref{l:evolution}.
 \begin{proof}
 Let $\varepsilon \in I$, $t \in \mathbb{R}$, and $\psi\in\mathcal{H}$, then
\begin{eqnarray*}
    \rme^{{\rmi}tH(\varepsilon)}S\rme^{-{\rmi}tH(\varepsilon)}-S&=\rme^{{\rmi}tH(\varepsilon)}\comm{S}{\rme^{-{\rmi}tH(\varepsilon)}}\\
    &=\rme^{{\rmi}tH(\varepsilon)}\comm{S}{\rme^{-{\rmi}t\tilde{H}(\varepsilon)}}+\rme^{{\rmi}tH(\varepsilon)}\comm{S}{\rme^{-{\rmi}tH(\varepsilon)}-\rme^{-{\rmi}t\tilde{H}(\varepsilon)}}\\&=
    A(t,\varepsilon)+B(t,\varepsilon).
\end{eqnarray*}
Moreover,
 \begin{eqnarray*}
        \norm{B(t,\varepsilon)\psi}&=\norm{\rme^{{\rmi}tH(\varepsilon)}\comm{S}{\rme^{-{\rmi}tH(\varepsilon)}-\rme^{-{\rmi}t\tilde{H}(\varepsilon)}}\psi}\\&=
        \norm{\comm{S}{\rme^{-{\rmi}tH(\varepsilon)}-\rme^{-{\rmi}t\tilde{H}(\varepsilon)}}\psi}\\&
        \leq
        \norm{S\left(\rme^{-{\rmi}t{H}(\varepsilon)}-\rme^{-{\rmi}t\tilde{H}(\varepsilon)}\right)\psi} +
        \norm{\left(\rme^{-{\rmi}t{H}(\varepsilon)}-\rme^{-{\rmi}t\tilde{H}(\varepsilon)}\right)S\psi}\\&\le
        \norm{S}\norm{\left(\rme^{-{\rmi}t{H}(\varepsilon)}-\rme^{-{\rmi}t\tilde{H}(\varepsilon)}\right)\psi}+
        \norm{\left(\rme^{-{\rmi}t{H}(\varepsilon)}-\rme^{-{\rmi}t\tilde{H}(\varepsilon)}\right)S\psi}.
    \end{eqnarray*}
Thus, by using Lemma~\ref{lemma:123},  we get the thesis~\eqref{eq:robust}.
\end{proof}

\subsection{Proof of Theorems}\label{susect:proof123}

We are now ready to prove the two main theorems of this section. We start with
Theorem~\ref{thm:main0}, and then we will prove Theorem~\ref{main}.

\begin{proof}[Proof of Theorem~\ref{thm:main0}] 
We want to prove that if $S$ is robust against the deformation $H(\varepsilon)$, namely if~\eqref{eq:robustS} holds, then $S$ commutes with $P_n(0)$ for all $n \geq 1$. We prove the negation of such implication, namely we assume that there exists $n \geq 1$ such that  $\comm{S}{P_{n}(0)}\neq0$ and  show that this implies that $S$ is fragile, i.e.~\eqref{eq:robustS} does not hold.

Since $\comm{S}{P_{n}(0)}\neq0$, there exists an integer $m\neq n$ such that $P_m(0) S P_n(0)\neq 0$ or $P_n(0) S P_m(0)\neq 0$. Let us assume, for definiteness, that the first inequality holds, and thus that there exist two unit vectors $\psi_n\in{P_n(0)}\mathcal{H}$ and $\psi_{m}\in{P_{m}(0)}\mathcal{H}$, such that $$\Braket{\psi_{m}| S\psi_{n}}\neq 0.$$

Let $\varepsilon \in I$ and $t \in \mathbb{R}$. By Lemma~\ref{l:evolution} we have that
\begin{eqnarray*}
    \norm{\left(\rme^{{\rmi}tH(\varepsilon)}S\rme^{-{\rmi}tH(\varepsilon)}-S\right)\psi_n}&=\norm{\left(A(t,\varepsilon)+B(t,\varepsilon)\right)\psi_n}\\
    &\ge \left|\norm{A(t,\varepsilon)\psi_n}-\norm{B(t,\varepsilon)\psi_n} \right|
\end{eqnarray*}
with 
$$\sup_{t \in \mathbb{R}}\, \norm{B(t,\varepsilon)\psi_n}\to 0, \quad \text{as } \varepsilon\to0.$$

On the other hand, by~\eqref{eq:Arob},
\begin{eqnarray*}
	\norm{A(t,\varepsilon)\psi_n}&=\norm{\rme^{{\rmi}t{H}(\varepsilon)}\comm{S}{\rme^{-{\rmi}t\tilde{H}(\varepsilon)}}\psi_n}	=\norm{\comm{S}{\rme^{-{\rmi}t\tilde{H}(\varepsilon)}}\psi_n}\nonumber\\
	&=\norm{\left(\rme^{{\rmi}t\tilde{H}(\varepsilon)}S\rme^{-{\rmi}t\tilde{H}(\varepsilon)}-S\right)\psi_n}\nonumber\\
	&=\norm{\left(\rme^{{\rmi}t\tilde{H}(\varepsilon)}S\rme^{-{\rmi}th_n(\varepsilon)}-S\right)\psi_n} \nonumber \\
	&=\sup_{\norm{\phi}=1}\abs{\Braket{\phi | \left(\rme^{{\rmi}t\tilde{H}(\varepsilon)}S\rme^{-{\rmi}th_{n}(\varepsilon)}-S\right)\psi_{n}}},
\label{eq:norm}
\end{eqnarray*}
where we made use of $\tilde{H}(\varepsilon)\psi_n=h_n(\varepsilon)\psi_n$.
By using the properties of the supremum and the equality $\tilde{H}(\varepsilon)\psi_m=h_m(\varepsilon)\psi_m$, we get
\begin{eqnarray*}
    \norm{A(t,\varepsilon)\psi_n}&\ge\abs{\Braket{\psi_{m} |\left(\rme^{{\rmi}th_{m}(\varepsilon)}S\rme^{-{\rmi}th_{n}(\varepsilon)}-S\right)\psi_{n}}}\\
    &=\abs{\rme^{{\rmi}t(h_{m}(\varepsilon)-h_{n}(\varepsilon))}-1}\abs{\Braket{\psi_{m}| S\psi_{n}}}\\
    &=2\abs{\sin{\left(\frac{t(h_{m}(\varepsilon)-h_{n}(\varepsilon))}{2}\right)}}\abs{\Braket{\psi_{m}| S\psi_{n}}}.
\end{eqnarray*}
Now,
\begin{equation*}
\sup_{t \in \mathbb{R}} \,\abs{\sin{\left(\frac{t(h_{m}(\varepsilon)-h_{n}(\varepsilon))}{2}\right)}}=1,
\end{equation*}
for all $\varepsilon\in I\setminus\{0\}$, by property $(ii)$ in Theorem~\ref{thm:main0}. 
Therefore,
\begin{equation*}
 \liminf_{\varepsilon\rightarrow0}    \sup_{t \in \mathbb{R}} \, \norm{\left(\rme^{{\rmi}tH(\varepsilon)}S\rme^{-{\rmi}tH(\varepsilon)}-S\right)\psi_n}\ge2\abs{\Braket{\phi_{m}|S\psi_{n}}} \neq 0,
\end{equation*}
and $S$ is fragile against the deformation $H(\varepsilon)$.

Now we prove the converse implication. Suppose that $\comm{S}{P_{n}(0)}=0$ for all $n\ge1$. Then we get that $A(t,\varepsilon)=0$ for all $\varepsilon \in I$ and $t \in \mathbb{R}$. Therefore, by using Lemma~\ref{l:evolution}, we have:
\begin{equation*}
    \sup_{t \in \mathbb{R}}\,\norm{\left(\rme^{{\rmi}tH(\varepsilon)}S\rme^{-{\rmi}tH(\varepsilon)}-S\right)\psi}=
    \sup_{t \in \mathbb{R}}\,\norm{B(t,\varepsilon)\psi}\to0, \quad \text{as } \varepsilon\to0,
\end{equation*}
for all $\psi\in\mathcal{H}$, i.e. $S$ is robust against the deformation $H(\varepsilon)$. 
\end{proof}

We end this subsection by finally proving Theorem~\ref{main}.
\begin{proof}[Proof of Theorem~\ref{main}]
The theorem is a corollary of Theorem~\ref{thm:main0}. 

In order to prove $(i)$, we will prove that $H(\varepsilon)=H +\varepsilon V$ is a continuous deformation of $H$ satisfying the assumptions $(i)$, $(ii)$ and $(iii)$ of Theorem~\ref{thm:main0}. This is a consequence of Kato's perturbation Theorem~\ref{th:KATO}. 

Indeed, $(i)$ is the spectral resolution given in Theorem~\ref{th:KATO}$(i)$. 

Moreover, by Theorem~\ref{th:KATO}$(ii)$, $\varepsilon\mapsto h_n(\varepsilon)$ are analytic for $\varepsilon\in (-1,1)$ and $h_n\neq h_m$ for $n\neq m$. Therefore, they are continuous and there exists an interval $I\subseteq (-1,1)$ with $0\in I$ such that $h_n(\varepsilon)\neq h_m(\varepsilon)$ for $\varepsilon\in I\setminus\{0\}$, and $(ii)$ of Theorem~\ref{thm:main0} is satisfied. 

It remains to prove that Theorem~\ref{thm:main0}$(iii)$ holds too. 
For all $\varepsilon \in (-1,1)$ consider the operator 
\begin{equation*}
W(\varepsilon)=\sum_{n \geq 1}P_n(\varepsilon)P_n(0),
\end{equation*}
implementing the following transformations between the spectral projections $P_m(\varepsilon)$ and $P_m(0)$:  for all $m \geq 1$
        \begin{equation*}
    W(\varepsilon)P_m(0)=P_m(\varepsilon)W(\varepsilon).
            \label{eq:basis_tran}
        \end{equation*}
In general $W(\varepsilon)$ is not a unitary operator. However, following  Kato~\cite[Secs.~4.6 and~6.8]{Kato_per}, we can define the operator  
\begin{equation*}
U(\varepsilon)= \sum_{n \geq 1} P_n(\varepsilon)P_n(0) (1-R_n(\varepsilon))^{-1/2},
\label{transformation}
\end{equation*}
where
\begin{equation*}
    R_n(\varepsilon)=(P_n(\varepsilon)-P_n(0))^2.
\end{equation*}
The operator $(1-R_n(\varepsilon))^{-1/2}$ is well defined for $\norm{R_n(\varepsilon)}<1$ and then for $\varepsilon$ sufficiently small. Then, there exists a neighbourhood $I\subseteq (-1,1)$ of 0  such that for all $\varepsilon\in I$ the operator $U(\varepsilon)$ is well defined.
An easy computation shows that $U(\varepsilon)$ is unitary and
that $P_m(\varepsilon)=U(\varepsilon)P_m(0)U(\varepsilon)^{\dagger}$  for all $\varepsilon \in I$ and all $m\geq1$.
   Finally, by the  continuity of $P_n(\varepsilon)$, the family $\{U(\varepsilon)\}_{\varepsilon \in I}$ is strongly continuous and thus property~$(iii)$ of Theorem~\ref{thm:main0} is verified. 
  
  Therefore, $H(\varepsilon)=H + \varepsilon V$ is a continuous deformation of $H$ and by Theorem~\ref{thm:main0} one gets $(i)$.

Now we prove $(ii)$. Since $V$ and $H$ commute and are both self-adjoint, we have that for all $\varepsilon \in I$ and $t\in \mathbb{R}$:
    \begin{equation*}
       \rme^{itH(\varepsilon)}= \rme^{it(H+\varepsilon{V})}=\rme^{{\rmi}tH}\rme^{{\rmi}t\varepsilon{V}}=\rme^{{\rmi}t\varepsilon{V}}\rme^{{\rmi}tH}.
    \end{equation*}
Hence for all $\psi \in \mathcal{H}$
  \begin{eqnarray*}
      \sup_{t\in \mathbb{R}}\,\norm{\left(\rme^{{\rmi}tH(\varepsilon)}S\rme^{-{\rmi}tH(\varepsilon)}-S\right)\psi}&=\sup_{t\in \mathbb{R}}\,\norm{\left(\rme^{{\rmi}t\varepsilon{V}}\rme^{{\rmi}tH}S\rme^{-{\rmi}tH}\rme^{-{\rmi}t\varepsilon{V}}-S\right)\psi} \nonumber \\
      &=\sup_{t\in \mathbb{R}}\,\norm{\left(\rme^{{\rmi}t\varepsilon{V}}S\rme^{-{\rmi}t\varepsilon{V}}-S\right)\psi}\nonumber \\
      &=\sup_{\tau \in \mathbb{R}}\norm{\left(\rme^{\rmi\tau{V}}S\rme^{-\rmi\tau{V}}-S\right)\psi}\nonumber \\
      &=\sup_{\tau \in \mathbb{R}}\norm{\left[\rme^{\rmi\tau{V}}, S\right]\psi} 
    \label{eq:commute}
  \end{eqnarray*}
  which is $\varepsilon$-independent and vanishes for all $\psi$ if and only if $V$ and $S$ commute.
  
{The last assertion is a direct consequence of $(i)$ and Remark~\ref{rmk:commutant}, by noting that in this case the bicommutant of $H$ is the von Neumann algebra generated by the eigenprojections $P_n$, and that $\{H,V\}'=\hat{\mathcal{R}}_{\{V\}}$ is the set of $V$-unbroken symmetries}.
\end{proof}

\begin{remark}
According to Definition~\ref{def:Pfr} a symmetry $S \in \{H\}'$ is $V$-fragile if
\begin{equation}
	 \limsup_{\varepsilon\rightarrow0}\, \sup_{t \in \mathbb{R}} \,\norm{\left(\rme^{{\rmi}t(H+\varepsilon V)}S\rme^{-{\rmi}t(H+\varepsilon V)}-S\right)\psi} > 0,
\end{equation}
for some $\psi\in\mathcal{H}$. However, looking at  the proof of Theorem~\ref{thm:main0}, one gets that  $S$ is $V$-fragile if and only if
\begin{equation}
	 \liminf_{\varepsilon\rightarrow0}\, \sup_{t \in \mathbb{R}}\, \norm{\left(\rme^{{\rmi}t(H+\varepsilon V)}S\rme^{-{\rmi}t(H+\varepsilon V)}-S\right)\psi} > 0,
\end{equation}
for some $\psi\in\mathcal{H}$.
\end{remark}

\section{$\mathcal{P}$-robustness and complete robustness}
\label{sec:4}
In Section~\ref{sec:3}, we have studied in detail the set of the symmetries robust against a single perturbation.  In this section we want to characterize the set of the symmetries robust against an arbitrary  set $\mathcal{P}$ of symmetric $H$-bounded perturbations. 

As a consequence of Theorem~\ref{main} we have the following result:
\begin{theorem}
\label{main2}
Let $\mathcal{P}$ be a set of symmetric $H$-bounded perturbations and $S\in\{H\}'$ be a symmetry of the Hamiltonian $H$. Then, $S$ is $\mathcal{P}$-robust if and only if 
\begin{equation}
[{S},{P_{n}^{(V)}(0)}]=0, \quad  \text{for all }  n \geq 1 \text{ and } V \in \mathcal{P},
\end{equation}
where $\{P^{(V)}_n(\varepsilon)\}_{n \geq 1}$ are the eigenprojections of the perturbed Hamiltonian $H+ \varepsilon V$. 
Therefore,
\begin{equation}
\mathcal{R}_{\mathcal{P}}(H)= \bigcap_{V \in \mathcal{P}} \mathcal{R}_{\{V\}}(H),
\end{equation}
with $\mathcal{R}_{\{V\}}(H)=\{P^{(V)}_n(0) \,:\, n\geq 1\}'$,
is a von Neumann algebra of bounded operators containing the bicommutant~$\{H\}''$.
\end{theorem}

As in the case of a single perturbation, the set of robust symmetries $\mathcal{R}_{\mathcal{P}}(H)$ forms a von Neumann algebra, thus it is closed with respect to the algebraic operations, as well as to taking the adjoint and taking strong limits. However, for a generic set $\mathcal{P}$ of perturbations
$\mathcal{R}_{\mathcal{P}}(H)$ is no longer the commutant of an Abelian algebra and has a more complex structure, which depends on $\mathcal{P}$. 

In the following two subsections we will look at the fine structure of the set of robust symmetries in the mathematical (and physical) interesting situation where the set of perturbations has itself a natural algebraic structure.

\subsection{Symmetry-restricted perturbations}

We want to analyze some interesting classes $\mathcal{P}$ of perturbations of the Hamiltonian~$H$. There are plenty of  physical situations where, by fundamental reasons, the dynamical laws of the system must be invariant under some symmetry group. Paradigmatic examples are given by the isotropy of space which implies rotational symmetry, or by the principle of relativity which implies relativistic invariance, etc.
Those protected symmetries restricts the class of possible dynamics and thus of possible Hamiltonians $H$ \emph{and} their  perturbations $\mathcal{P}$.

Let us  consider $n$ protected Hermitian symmetries $J_1,J_2,\dots,J_n\in\{H\}'$, that in general are not assumed to commute among themselves. In the above physical examples, these can be, e.g., the generators of the rotation group, or, in relativistic systems, of the Lorentz group. 

What happens when the set of perturbations is made by the Hermitian operators of $\{J_1,J_2,\dots, J_n\}'$? This is exactly the largest set of perturbations that leave the above $n$ symmetries unbroken. We are going to prove that the set $\mathcal{R}_\mathcal{P}(H)$ of robust symmetries against this kind of perturbations has a very nice algebraic structure.

We first consider the situation of a single protected symmetry. Let $J\in\{H\}'$ be a Hermitian symmetry of  $H$ and consider as set of perturbations $\mathcal{P}=\{J\}'_h$, the Hermitian elements of the commutant of $J$, which is the largest set of perturbations which keep $J$ unbroken. This would be, for example, the case of a system with cylindrical symmetry where, say, $J=J_z$ is conserved. In such a case the set of robust symmetries  (against this kind of perturbations) acquires the algebraic structure of a bicommutant.

\begin{lemma}
\label{thm:symmetry}
Let $H$ be self-adjoint with compact resolvent and $J\in\{H\}'$ be a  Hermitian symmetry of $H$. Let  $\mathcal{P}=\{J\}'_h$, where $\{J\}'_h$ are the Hermitian elements of $\{J\}'$. Then, a symmetry $S \in\{H\}'$ is $\mathcal{P}$-robust if and only if $S\in\{H,J\}''$, 
 where
        \begin{equation}
       \{H,J\}''
        = \{Q_{k} \,:\, k\geq 1 \}''
        \label{bicommutantA}
    \end{equation}
    is the bicommutant of  $H$ and $J$, with $Q_{k}$ being the common eigenprojections of $H$ and~$J$. Therefore,
     \begin{equation}
     	\mathcal{R}_\mathcal{P}(H)=\{H,J\}''.
     \end{equation}
\end{lemma}
\begin{proof} 
We start by proving that if $S$ is a $\mathcal{P}$-robust symmetry then $S \in \{H,J\}''$. 

Notice first that, by von Neumann's bicommutant theorem
\begin{equation}
	\{H,J\}''
        = \{Q_{k} \,:\, k\geq 1 \}''
        = \Bigl\{\sum_{k \geq 1} a_{k} Q_{k}      
        \,:\, (a_{k})_{k \geq 1}\in l^{\infty}
        \Bigr\},
 \label{eq:bicommutantA1}
\end{equation}
where the convergence of the series is in the strong topology.  Consider perturbations in $V=V^\dag \in \{H,J\}'$. Then $V\psi=\sum_{k\ge 1}Q_k V Q_k\psi$, for all $\psi\in\mathcal{H}$. Since $V$ commutes with $H$, by Theorem~\ref{main}$(ii)$, we get that
    \begin{eqnarray}
       \left[S,V\right]=0, \quad \text{for all } V=V^\dag \in \{H,J\}'.
    \end{eqnarray}
   By choosing $V= \sum_{k\geq 1} (1/k) Q_k$ we get
   \begin{equation*}
        S=\sum_{k\ge1}Q_k S Q_k,
    \end{equation*}
    in the strong topology.    
    By choosing for any $k\geq 1$ the perturbation in the form $V= Q_k W Q_k$ with arbitrary $W=W^\dag$ we get
\begin{equation*}
        \left[Q_k S Q_k, Q_k W Q_k\right]=0, \quad \text{for all } W=W^\dag \in B(\mathcal{H}),
    \end{equation*}
 which gives, according to Schur's lemma, $Q_k S Q_k = s_k Q_k$, for some $s_k\in\R$ with $|s_k|\leq \|S\|$.
 We conclude that
    \begin{equation*}
        S=\sum_{k\ge1}s_k Q_k \in\{H,J\}''.
    \end{equation*}

Now we prove that if $S\in\{H,J\}''$ then $S$ is a $\mathcal{P}$-robust symmetry, with $\mathcal{P}=\{J\}'_h$.
According von Neumann's bicommutant theorem~\eqref{eq:bicommutantA1}, 
there is a sequence $(s_{k})_{k \geq 1} \in \ell^{\infty}$ such that
    \begin{equation*}
        S=\sum_{k\geq 1} s_{k}Q_{k},
    \end{equation*}
in the strong topology. 
   
 Let us fix a perturbation $V\in\{J\}'_h$. 
Since both $H$ and $V$ commute with $J$, then the same has to be true for $H(\varepsilon)=H+\varepsilon{V}$, with $\varepsilon\in{I}$. This clearly means that 
\begin{equation*}
    \comm{Q_{k}}{P_n(\varepsilon)}=0,
\end{equation*}
for $\varepsilon\in{I}$, $\varepsilon\neq{0}$ and $n,k\ge{1}$. By taking the limit $\varepsilon\rightarrow0$ and by using the continuity of $P_n(\varepsilon)$, we get
\begin{equation*}
    \comm{Q_{k}}{P_n(0)} =0.
    \label{eq:commut}
\end{equation*}
Thus, for all $n \geq 1$ and $\psi\in \mathcal{H}$,
    \begin{eqnarray*}
        \comm{S}{P_{n}(0)}\psi=\sum_{k\ge1}s_{k}\comm{Q_{k}}{P_{n}(0)}\psi=0,
        \label{bicomm}
    \end{eqnarray*}
   so that  $S$ is $V$-robust by Theorem~\ref{main}$(ii)$. 
By the arbitrariness of $V$ we get the thesis.
\end{proof}

We are now ready to tackle the generic situation of an arbitrary family of protected symmetries.

\begin{theorem}
\label{thm:restricted}
    Let $H$ be self-adjoint with compact resolvent. Let $\mathcal{P}= \mathcal{J}'_h$, with $\mathcal{J}\subseteq\{H\}'_h$ being a family of Hermitian protected symmetries of $H$.  Then
    \begin{equation}
        \mathcal{R}_\mathcal{P}(H)=\left(\{H\}\cup\mathcal{J}\right)''.
    \end{equation}
\end{theorem}

\begin{proof}
    We are going to prove the double inclusion. First of all let us prove that $\mathcal{R}_\mathcal{P}(H)\supseteq \left(\{H\}\cup\mathcal{J}\right)''$.
    
    Let $\mathcal{P}_J=\{J\}'_h$ for $J\in\mathcal{J}$.   Since $\mathcal{P}=\mathcal{J}'_h=\bigcap_{J\in\mathcal{J}}\mathcal{P}_J$, every  $\mathcal{P}_J$-robust symmetry is also $\mathcal{P}$-robust, that is
    \begin{equation*}
    	\bigcup_{J\in\mathcal{J}} \mathcal{R}_{\mathcal{P}_J}(H)\subseteq\mathcal{R}_\mathcal{P}(H). 
    \end{equation*}
 Now, by the previous theorem we get that $\mathcal{R}_{\mathcal{P}_J}(H)=\{H,J\}''$, whence
    \begin{equation*}
       \bigcup_{J\in\mathcal{J}} \{H,J\}''\subseteq\mathcal{R}_\mathcal{P}.
    \end{equation*}
    Taking the bicommutant of both members, considering that $\mathcal{R}_\mathcal{P}$ is a von Neumann algebra, i.e. $\mathcal{R}_\mathcal{P}''=\mathcal{R}_\mathcal{P}$ , we get
    \begin{equation*}
        \Bigl(\bigcup_{J\in\mathcal{J}} \{H,J\}''\Bigr)''\subseteq\mathcal{R}_\mathcal{P}.
    \end{equation*}
    By using the properties of the bicommutant we can compute the left hand side
    \begin{equation*}
       \Bigl(\bigcup_{J\in\mathcal{J}} \{H,J\}''\Bigr)'' = \Bigl(\bigcup_{J\in\mathcal{J}} \{H,J\}\Bigr)''= \left(\{H\}\cup\mathcal{J}\right)''.
    \end{equation*}
    Then we have proved that
    \begin{equation*}
        \left(\{H\}\cup\mathcal{J}\right)''\subseteq\mathcal{R}_\mathcal{P}.
    \end{equation*}
    
Now we are going to prove that $\mathcal{R}_\mathcal{P}\subseteq \left(\{H\}\cup\mathcal{J}\right)''$. Consider $S\in\mathcal{R}_\mathcal{P}$. We have to show that $\comm{S}{B}=0$, for all ${B\in \left(\{H\}\cup\mathcal{J}\right)'}$. Let us decompose $B$ as
\begin{equation*}
    B=V_1+\rmi{V_2},
\end{equation*}
where
\begin{equation*}
    V_1=\frac{B+B^\dagger}{2}\quad V_2=\frac{B-B^\dagger}{2\rmi}.
\end{equation*}
Since $H$ and $J\in\mathcal{J}$ are self-adjoint operators, $\left(\{H\}\cup\mathcal{J}\right)'$ is a von Neumann algebra. Then since $B\in \left(\{H\}\cup\mathcal{J}\right)'$, also $B^\dagger\in \left(\{H\}\cup\mathcal{J}\right)'$ and then $V_i\in\left(\{H\}\cup\mathcal{J}\right)'$, for $i=1,2$. Then, since $S$ is $\mathcal{P}$-robust and $V_1, V_2\in \mathcal{P}$ commute with $H$, by Theorem~\ref{main}$(ii)$ we have that
\begin{equation*}
    \comm{S}{V_1}=\comm{S}{V_2}=0,
\end{equation*}
that is $\comm{S}{B}=0$.    
Therefore,
    \begin{equation*}
        \mathcal{R}_\mathcal{P}\subseteq \left(\{H\}\cup\mathcal{J}\right)''	,
    \end{equation*}
and the theorem is proved. 
\end{proof}

\subsection{Completely robust symmetries}

Finally, we provide a characterization of the completely robust symmetries of $H$, that are robust against unrestricted perturbations. 
\begin{theorem}
\label{complete}
Let $H$ be self-adjoint with compact resolvent. Let $S\in\{H\}'$ be a symmetry of $H$. Then $S$ is a completely robust symmetry if and only if $S\in\{H\}''$, 
 where
$\{H\}''$ is the bicommutant~\eqref{eq:bicommutantCR} of the Hamiltonian $H$. Therefore,
\begin{equation}
     \mathcal{R}(H)=\{H\}''.
\end{equation}
\end{theorem}
\begin{proof}
    It follows from Lemma~\ref{thm:symmetry} with $J=\mathbb{I}$.
\end{proof}

The characterization of completely robust symmetries as the elements of the bicommutant of $H$, that are the bounded functions of the Hamiltonian, was proved in~\cite{kam} for finite-dimensional quantum systems. Theorem~\ref{complete} generalizes this result also to unbounded Hamiltonians with compact resolvent in an infinite-dimensional Hilbert space. 

The extension of the validity of this characterization to unbounded Hamiltonians could have been anticipated. Indeed, heuristically, the long-time dynamics is expected to be unaffected by the ultraviolet (i.e.\ high energy) behavior of the Hamiltonian. And the assumption of compact resolvent, which allows the application of Kato's perturbation theorem and ensures the stability of the pure-point spectrum, is an infrared (i.e.\ low energy) condition on the spectrum of $H$.

\section{Wandering range and adiabatic invariants}
\label{sec:adiabatic}

Before concluding this work we comment on some additional properties of robust symmetries that would be worth understanding better  and that provide a perspective for future investigations and some interesting connections with stability theory in classical Hamiltonian mechanics.

\subsection{Wandering range of a robust symmetry}
According to Definition~\ref{def:Pfr}, a symmetry $S$ is $V$-robust if, for all $\varphi\in\mathcal{H}$:
\begin{equation}
\label{eq:convergence}
   \delta_{H+\varepsilon V} (S; \varphi) =\sup_{t \in \R} \,\norm{(\rme^{it(H+\varepsilon V)}S\rme^{-it(H+\varepsilon V)}-S)\varphi} \to 0, \quad \text{as } \varepsilon\to0.
\end{equation}
Here we want to discuss the speed of convergence of the above limit.

For this purpose it would be useful to have an explicit upper bound of the divergence $\delta_{H+\varepsilon V} (S; \varphi)$, which physically represents the \textit{wandering range} of the perturbed evolution of the symmetry $\rme^{it(H+\varepsilon V)}S\rme^{-it(H+\varepsilon V)}$ around its unperturbed evolution $S$.

For finite-dimensional systems the following uniform bound for robust symmetries holds:
\begin{equation}
	\sup_{t \in \R} \,\norm{\rme^{it(H+\varepsilon V)}S\rme^{-it(H+\varepsilon V)}-S} \leq  \frac{14 \sqrt{d} \|V\|\|S\|}{\eta}  |\varepsilon|,
\end{equation}
where $d$ is the number of distinct eigenvalues of $H$ and $\eta$ is its minimal spectral gap~\cite{kam}. 

A natural question arises from the bound obtained for finite-dimensional systems: given a $V$-robust symmetry $S \in \{H\}'$ and a unit vector $\varphi \in \mathcal{H}$, one can ask whether it is possible to find a constant $C_\varphi >0$ and $\varepsilon_\varphi^* \in (0,1)$ such that for all $\varepsilon \in (-\varepsilon_\varphi^*, \varepsilon_\varphi^*)$:
\begin{equation}\label{speed1}
\sup_{t \in \R}\, \norm{(\rme^{{\rmi}t(H+\varepsilon V)}S\rme^{-{\rmi}t(H+\varepsilon V)}-S) \varphi} \leq{C_\varphi}\abs{\varepsilon}.
\end{equation}
Or, even more, whether the above bound holds in operator norm, i.e. for positive $C$ and $\varepsilon^*$ independent of $\varphi$.
If this were the case, one would have that the wandering range $\delta_{H+\varepsilon V} (S; \varphi)$ of the robust symmetry $S$ is of
$O(\varepsilon)$ (and uniformly in $\varphi$ for a bound in operator norm).

The answer to the above question is negative (if $\mathcal{H}$ is infinite-dimensional), as shown in the following example: in general the convergence~\eqref{eq:convergence} is  not uniform in $\varphi$ (with $\|\varphi\|=1$) and hence does not hold in operator norm. Moreover, there are vectors $\varphi$ such that the wandering range is of $O(|\varepsilon|^\gamma)$ with $\gamma>0$ arbitrarily small. This is reminiscent of the phenomenon of Arnold diffusion in classical mechanics~\cite{Arnold_diffusion}.

\begin{example}
Let $H$ be the Hamiltonian of a one-dimensional harmonic oscillator with mass $m=1$ and frequency $\omega=1$ on the Hilbert space $L^2(\R)$, 
\begin{equation}
    H=\frac{1}{2}\left(\hat{p}^2+\hat{x}^2\right),
\end{equation}
with domain $D(H)\subset H^2(\R)$, where $\hat{p}=-i \frac{d}{dx}$ is the momentum operator and $\hat{x}$ is the position operator. 
The Hamiltonian $H$ has compact resolvent and its spectrum is discrete with simple eigenvalues
\begin{equation}
    h_n=n+\frac{1}{2},\quad n \in \mathbb{N} .
\end{equation}
Let us consider the self-adjoint $H$-bounded perturbation
\begin{equation}
    V=\hat{p},
\end{equation}
with domain $D(V)=H^1(\R)$.
For all $\varepsilon \in \R$ the perturbed Hamiltonian reads
\begin{equation}
    H(\varepsilon)=H+ \varepsilon V=\frac{1}{2}\left(\hat{p}^2+\hat{x}^2+2\varepsilon \hat{p}\right)=\frac{1}{2}\left((\hat{p}+\varepsilon )^2+\hat{x}^2\right)-\frac{\varepsilon^2}{2} , 
\end{equation}
with domain $D(H(\varepsilon))=D(H)$, that is the Hamiltonian of a harmonic oscillator with shifted momentum, with spectrum
\begin{equation}
    h_n(\varepsilon)=h_n-\frac{\varepsilon^2}{2}, \quad n \in\mathbb{N}.
\end{equation}
An easy computation shows that the family of unitary operators $\{U(\varepsilon)\}_{\varepsilon \in \R}$ in Lemma~\ref{lemma:123} is given by
\begin{equation}
    U(\varepsilon)=\rme^{{\rmi\varepsilon}\hat{x}}, \quad \varepsilon \in \R,
\end{equation}
and that the eternal block-diagonal approximation~(\ref{eqn:Htildee}) reads
\begin{equation}
    \tilde{H}(\varepsilon)=e^{-\rmi\varepsilon \hat{x}}H(\varepsilon)e^{\rmi\varepsilon\hat{x}}=H-\frac{\varepsilon^2}{2} .
\end{equation}
This means that $H$ and $\tilde{H}(\varepsilon)$ share the eigenprojections for all $\varepsilon \in \R$, hence all the symmetries of $H$ are $V$-robust, namely for all $S \in \{H\}'$ and for all $\psi \in L^2(\R)$,
\begin{equation}\label{limitepsto0}
 \delta_{H(\varepsilon)} (S; \psi) =\sup_{t \in \R}\, \norm{(\rme^{{\rmi}tH(\varepsilon)}S\rme^{-{\rmi}tH(\varepsilon)}-S)\psi}\to 0, \quad \text{as } \varepsilon\to 0.
\end{equation}

We want to construct a (robust) symmetry $S$ of $H$ and a vector $\psi$ of the Hilbert space $L^2(\R)$ such that~(\ref{speed1}) is not true.
Consider the following symmetry of $H$
\begin{equation}
    S=\frac{1-\Pi}{2},
\end{equation}
where $\Pi$ is the parity operator, $\Pi \psi(x) =\psi(-x)$, for $\psi\in L^2(\R)$. Notice that in fact $S$ is  completely robust, since $\Pi\in\{H\}''$ (see Theorem~\ref{complete}), being a bounded function of the Hamiltonian $H$: 
\begin{equation}
	\Pi =\rmi \rme^{-\rmi{\pi}H}.
	\label{eq:PiH}
\end{equation}

Consider the wave function
\begin{equation}
    \psi_{\alpha}(x)=\frac{1}{(1+x^2)^{\alpha/4}},\quad x \in \R,
\end{equation}
with $\alpha >1$. Since $S\psi_{\alpha}=0$, we have that for all $\varepsilon \in \R$:
\begin{eqnarray}
    \sup_{t \in \R} \, \norm{\left(\rme^{{\rmi}tH(\varepsilon)}S\rme^{-{\rmi}tH(\varepsilon)}-S\right)\psi_{\alpha}}&=\sup_{t \in \R} \,\norm{S\rme^{-{\rmi}tH(\varepsilon)}\psi_{\alpha}}\nonumber\\
    &=\sup_{t \in \R} \,\norm{S \rme^{\rmi\varepsilon \hat{x}}\rme^{{-\rmi}t(H-{\varepsilon^2}/{2})}\rme^{-\rmi\varepsilon \hat{x}}\psi_{\alpha}}\nonumber\\
    &=\sup_{t \in \R} \,\norm{S\rme^{\rmi\varepsilon \hat{x}}\rme^{{-\rmi}tH}\rme^{-\rmi\varepsilon \hat{x}}\psi_{\alpha}}\nonumber\\
    &\geq\norm{S \rme^{\rmi\varepsilon \hat{x}}\rme^{{-\rmi}\pi H}\rme^{-\rmi\varepsilon \hat{x}}\psi_{\alpha}}\nonumber \\
    &= \norm{-\rmi S \rme^{\rmi\varepsilon \hat{x}}\Pi\,\rme^{-\rmi\varepsilon \hat{x}}\psi_{\alpha}},
\end{eqnarray}
where in the last equality we used~\eqref{eq:PiH}.
Since for all $\varepsilon \in \R$
\begin{equation}
 \left(-\rmi{S\rme^{\rmi\varepsilon \hat{x}}\Pi\rme^{-\rmi\varepsilon\hat{x}}\psi_{\alpha}}\right)(x)=
  \sin{\left(2\varepsilon x\right)}\psi_{\alpha}(x), \quad x \in \R,
\end{equation}
we have that 
\begin{eqnarray}
    \norm{-\rmi S\rme^{\rmi\varepsilon \hat{x}}\Pi\rme^{-\rmi\varepsilon \hat{x}}\psi_{\alpha}}^2&=\int_{\mathbb{R}}{\sin^2{\left(2\varepsilon{x}\right)}}\abs{\psi_{\alpha}(x)}^2\rmd{x}=\int_{\mathbb{R}}\frac{\sin^2{\left(2\varepsilon{x}\right)}}{\left(1+x^2\right)^{\frac{\alpha}{2}}}\rmd{x}\nonumber \\
    &=\abs\varepsilon^{\alpha-1}\int_{\mathbb{R}}\frac{\sin^2{\left(2{x}\right)}}{\left(\varepsilon^2+x^2\right)^{\frac{\alpha}{2}}}\rmd{x}
    \nonumber\\
    &\ge\abs\varepsilon^{\alpha-1}\int_{\mathbb{R}}\frac{\sin^2{\left(2{x}\right)}}{\left(1+x^2\right)^{\frac{\alpha}{2}}}\rmd{x},
\end{eqnarray}
where the last inequality is true for $\abs\varepsilon<1$. Therefore, for all $\varepsilon \in (-1,1)$:
\begin{equation}
    \delta_{H(\varepsilon)} (S; \psi_\alpha) =\sup_{t \in \R} \,\norm{\left(\rme^{{\rmi}tH(\varepsilon)}S\rme^{{-\rmi}tH(\varepsilon)}-S\right)\psi_{\alpha}}\ge{c_\alpha}|\varepsilon|^{\frac{\alpha-1}{2}},
\end{equation}
where
\begin{equation}
    c_\alpha=\left(\int_{\mathbb{R}}\frac{\sin^2{\left(2{x}\right)}}{\left(1+x^2\right)^{\frac{\alpha}{2}}}d{x}\right)^{\frac{1}{2}} < + \infty,
\end{equation}
because $\alpha>1$. Therefore, (\ref{speed1}) cannot hold if $\alpha \in (1,3)$. In fact, the speed of convergence can be arbitrarily slow for $\alpha\approx 1$.

This example also shows that the convergence cannot be in operator norm. Indeed, since $c_\alpha / \|\psi_\alpha\| \to  1/\sqrt{2}$, as $\alpha\downarrow 1$, then  
\begin{eqnarray}
 \fl\qquad\qquad   \sup_{t \in \R} \,\norm{\rme^{{\rmi}tH(\varepsilon)}S\rme^{{-\rmi}tH(\varepsilon)}-S}&\ge& \norm{\rme^{{\rmi}\pi H(\varepsilon)}S\rme^{{-\rmi}\pi H(\varepsilon)}-S}
    \nonumber\\
    &\ge& 
    \sup_{\alpha > 1} \frac{\norm{\left(\rme^{{\rmi}\pi H(\varepsilon)}S\rme^{{-\rmi}\pi H(\varepsilon)}-S\right)\psi_{\alpha}}}{\|\psi_\alpha\|}
    \ge
     \frac{1}{\sqrt{2}}.
\end{eqnarray}

\end{example}
\subsection{Quantum adiabatic invariants}
\label{sec:5}

When the Hamiltonian of a quantum system is continuously deformed, a symmetry is in general no longer a conserved quantity with respect to the perturbed dynamics. However, if a symmetry is robust against the perturbation, it is possible to continuously deform it and obtain a new symmetry of the perturbed dynamics, called an \emph{adiabatic invariant}.
\begin{theorem}
\label{thm:adinv}
Let $H(\varepsilon)$ be a continuous deformation of $H$ as in Theorem~\ref{thm:main0} and
let $\{U(\varepsilon)\}_{\varepsilon \in I}$ be the corresponding strongly continuous family of unitary operators. Let $S \in \{H\}'$ be a robust symmetry against the deformation $H(\varepsilon)$. We define for all $\varepsilon \in I$ 
\begin{equation}
    S_\varepsilon:=U(\varepsilon)SU(\varepsilon)^\dagger.
        \label{eq:rotation}
\end{equation}
Then $S_\varepsilon$ is a symmetry for the system with Hamiltonian
$H(\varepsilon)$, i.e. $S_\varepsilon \in \{H(\varepsilon)\}'$, for all $\varepsilon \in I$. 
\end{theorem}
\begin{proof} Let $\varepsilon \in I$. Since $S$ is a robust symmetry, according to Theorem~\ref{thm:main0}, $S$ commutes with all the projections $\{P_n(0)\}_{n \geq 1}$, hence $S$ commutes with $\tilde{H}(\varepsilon)$ in (\ref{eqn:Htildee}). Moreover, since $H(\varepsilon)=U(\varepsilon)\tilde{H}(\varepsilon)U(\varepsilon)^\dagger$, then
\begin{eqnarray*}
    \rme^{{\rmi}tH(\varepsilon)}S_\varepsilon{\rme^{-{\rmi}tH(\varepsilon)}}&=\rme^{{\rmi}tH(\varepsilon)}U(\varepsilon)SU(\varepsilon)^\dagger{\rme^{-{\rmi}tH(\varepsilon)}}\\
    &=U(\varepsilon)\rme^{{\rmi}t\tilde{H}(\varepsilon)}S\rme^{-{\rmi}t\tilde{H}(\varepsilon)}U(\varepsilon)^\dagger \\
    &=U(\varepsilon)SU(\varepsilon)^\dagger \\
    &=S_\varepsilon.
\end{eqnarray*}
\end{proof}
\begin{remark}
Equation~\eqref{eq:rotation} is just one possible definition of an adiabatic invariant. Indeed, it is enough to define for all $\varepsilon \in I$
\begin{equation}
    S_\varepsilon:=U(\varepsilon)\bar{S}_\varepsilon{U(\varepsilon)}^\dagger,
\end{equation}
where $\bar{S}_\varepsilon$ is any operator in ${\{\tilde{H}(\varepsilon)\}'}$ such that
\begin{eqnarray}
    \bar{S}_\varepsilon \psi \to S \psi, \quad \text{as }\varepsilon\to 0,\qquad \forall \psi \in \mathcal{H}.
\end{eqnarray}

Such construction works if and only if $S$ is a robust symmetry. This suggests an alternative definition of a robust symmetry: $S$ is robust against the continuous deformation $H(\varepsilon)$ if it  can be continuously deformed in such a way it is conserved with respect to the perturbed dynamics generated by $H(\varepsilon)$.

From this point of view a robust symmetry can be viewed as a symmetry which is not broken by the perturbations, but instead is only bent by them.
In classical mechanics  this is indeed the way in which one defines the invariant KAM tori~\cite{Arnold}. 
\end{remark}

\section{Conclusions}
\label{sec:conclusions}

In this work we have seen that it is possible, for dynamics generated by a compact-resolvent Hamiltonian $H$, to give a complete algebraic characterization of the symmetries robust against an arbitrary set $\mathcal{P}$ of perturbations of $H$. If a generic symmetry is characterized by the fact that it commutes with the eigenprojections of the Hamiltonian, the $V$-robust symmetries are the ones that commute with the subprojections induced by~$V$. The set of $\mathcal{P}$-robust symmetries, for an arbitrary $\mathcal{P}$ are then obtained as the intersection of the sets of  symmetries robust against a single perturbation. We have studied, in particular, the case in which the set of perturbations is the commutant of a set of Hermitian symmetries of the unperturbed Hamiltonian. We have seen that in such a case the set of robust symmetries is the von Neumann algebra generated by the Hamiltonian and this set of Hermitian symmetries. Then, we proved that the set of completely robust symmetries is given by the bicommutant of the Hamiltonian, generalizing the result of~\cite{kam} to unbounded Hamiltonians with compact resolvent in infinite-dimensional systems.

Notice that the hypothesis of compact resolvent of the Hamiltonian is  instrumental in deriving all our results. More precisely, it is essential that the family of operators $H(\varepsilon) := H + \varepsilon V$ has discrete spectrum for $\varepsilon$ sufficiently small, as it is clear from Theorem~\ref{thm:main0}. However, unless $H$ has compact resolvent, this is not in general true, even though $H$ has a discrete spectrum. 

It could be interesting to see what happens for Hamiltonians $H$ with discrete spectrum but with a vanishing gap, and/or with a continuum spectrum, as for instance happens in the case of the hydrogen atom. Clearly, additional conditions on the perturbations $V$ will be needed in order to ensure the stability of the spectrum of the Hamiltonian~\cite{poschel,discrete,gapless}.
Here, a quantum version of the KAM theory~\cite{Scherer,Scherer2,Kuksin_2013}, which reproduces the superconvergent scheme used by Kolmogorov in the classical case, might be useful.
Finally, it would be useful to find explicit bounds for the wandering range of a robust symmetry, in order to better  understand  its stability against a perturbation in infinite-dimensional quantum systems.

\ack
We acknowledge support from the Italian National Group of Mathematical Physics (GNFM-INdAM) and from PNRR MUR projects CN00000013-``Italian National Centre on HPC, Big Data and Quantum Computing'' and PE0000023-NQSTI. 
PF and VV acknowledge support from INFN through the project ``QUANTUM'' and from the Italian funding within the ``Budget MUR - Dipartimenti di Eccellenza 2023--2027''  - Quantum Sensing
and Modelling for One-Health (QuaSiModO).

\newcommand{\newblock}{}
\bibliographystyle{iopart-num}

\bibliography{iopart-num.bib}
\end{document}